\documentclass[11pt,english]{article}

\usepackage[utf8]{inputenc}
\usepackage{babel}
\usepackage{amsmath}
\usepackage{amsthm}
\usepackage[unicode=true,pdfusetitle,
 bookmarks=true,bookmarksnumbered=false,bookmarksopen=false,
 breaklinks=false,pdfborder={0 0 1},backref=false,colorlinks=false]
 {hyperref}

\makeatletter
\numberwithin{equation}{section}
\numberwithin{figure}{section}
\theoremstyle{plain}
\newtheorem{lem}{\protect\lemmaname}[section]
\ifx\proof\undefined\
  \newenvironment{proof}[1][\proofname]{\par
    \normalfont\topsep6\p@\@plus6\p@\relax
    \trivlist
    \itemindent\parindent
    \item[\hskip\labelsep
          \scshape
      #1]\ignorespaces
  }{%
    \endtrivlist\@endpefalse
  }
  \providecommand{\proofname}{Proof}
\fi
\theoremstyle{plain}
\newtheorem{cor}{\protect\corollaryname}[section]
\theoremstyle{plain}
\newtheorem{thm}{\protect\theoremname}[section]

\usepackage[hmargin=1in,bmargin=1.35in,tmargin=0.5in]{geometry}
\date{}
\hypersetup{
colorlinks=true,     linkcolor=blue,     citecolor=blue,     urlcolor=blue }

\makeatother

\usepackage[style=alphabetic]{biblatex}
\providecommand{\corollaryname}{Corollary}
\providecommand{\lemmaname}{Lemma}
\providecommand{\theoremname}{Theorem}

\begin{document}
\title{A Tale of Two Trees: New Analysis for AVL Tree and Binary Heap}
\author{Russel L. Villacarlos\thanks{Department of Information Technology, Cavite State University. Email:
\texttt{rlvillacarlos@cvsu.edu.ph}. Research conducted at University
of the Philippines - Los Baños with support from the Accelerated Science
and Technology Human Resource Development Program of the Department
of Science and Technology.}\and Jaime M. Samaniego\thanks{Institute of Computer Science, University of the Philippines - Los
Baños. Email: \texttt{jmsamaniego2@up.edu.ph}.}\and Arian J. Jacildo\thanks{Institute of Computer Science, University of the Philippines - Los
Baños. Email: \texttt{ajjacildo@up.edu.ph}.}\and Maria Art Antonette D. Clariño\thanks{Institute of Computer Science, University of the Philippines - Los
Baños. Email: \texttt{mdclarino@up.edu.ph}.}}
\maketitle
\begin{abstract}
In this paper, we provide new insights and analysis for the two elementary
tree-based data structures -- the AVL tree and binary heap. We presented
two simple properties that gives a more direct way of relating the
size of an AVL tree and the Fibonacci recurrence to establish the
AVL tree’s logarithmic height. We then give a potential function-based
analysis of the bottom-up heap construction to get a simpler and tight
bound for its worst-case running-time.
\end{abstract}

\section{Introduction}

The AVL tree \cite{AdelsonVelskii1962} and binary (max-)heap \cite{Williams1964,Cormen2009}
are arguably the most elementary tree data structures in the literature.
The AVL tree is the first binary search tree data structure with guaranteed
logarithmic height. Proving its logarithmic height relies on bounding
$N(h)$, the minimum number of nodes needed to construct an AVL tree
with height $h$. A tree of minimum size must have subtrees of different
height and of minimum size thus, $N(h)=N(h-1)+N(h-2)+1$.

The structure of the recurrence suggests that $N(h)$ is related to
the \textit{Fibonacci recurrence}. Indeed, it is provable via induction
that $N(h)=F(h+2)-1$, where $F(i)$ is the $i^{th}$ Fibonacci number.
Since it is known that $F(i)\ge\phi^{(i-2)}$ , where $\phi=(1+\surd5)/2$
is the \textit{golden ratio}, it follows that $h=O(\log_{\phi}n)$
with $n$ being the number of elements in the tree. 

There has been renewed interest in AVL tree with the development of
new AVL variants -- \textit{rank-balanced tree} \cite{Haeupler2009}
and \textit{ravl tree} \cite{Sen2010}, and the analysis of AVL tree
performance with respect to the number of rotations \cite{Amani2016}.
Very recently, an open problem posed in \cite{Grenet2020} asks for
an alternative analysis of the height of AVL tree using a potential
function. In this paper, we partially address this problem by presenting
an alternative analysis. Although not based on a potential function
argument, our analysis offers new insights on properties of AVL trees.

The binary heap is a simple data structure primarily used for the
implementation of the heapsort algorithm \cite{Williams1964} and
priority queue \cite{Cormen2009}. An important algorithm for binary
heap, introduced by Floyd \cite{Floyd1964}, efficiently constructs
a heap from an $n$ element array. The algorithm, which we refer to
as \textit{build-heap}, treats the input array as an ordinary binary
tree then applies, bottom-up, the \textit{sift-down} operation to
each subtree. This sift-down operation transforms a subtree into a
heap by moving down the root element, via exchanges with a child,
until it is in its proper position in the subtree. 

The running-time of build-heap can be obtained by summing-up the time
taken by the individual sift-down operations. A sift-down, in the
worst-case, takes $O(h)$ time, where $h$ is the height of the subtree.
This is since the root element can go down the very bottom of the
subtree. There are $\lceil n/2^{(h+1)}\rceil$ subtrees of height
$h$ and the maximum height of a subtree is $\log_{2}n$. Therefore,
the running time can be described by the sum ${\displaystyle \sum_{h=0}^{\log_{2}n}}h\cdot\lceil n/2^{(h+1)}\rceil$,
which is at most $2n$. 

The analysis of build-heap presented above is a form of \textit{aggregation}
commonly used in amortized analysis \cite{Cormen2009}. Apart from
aggregation, another common technique used is the \textit{potential
method}. The recent result in \cite{Grenet2020} used this potential
method to obtain a simple analysis of the \textit{Euclidean} algorithm.
We take a similar approach in this paper for a simpler analysis of
the running-time of build-heap. 

\subsection{Our Contributions}

In this paper, we provide alternative analysis for the height of AVL
tree and the running time of the build-heap algorithm. In our analysis
of AVL tree we present two simple properties of AVL tree and used
them to directly prove that $N(h)=F(h+2)-1$. The first property shows
that an AVL tree with $N(h)$ elements can be constructed from a tree
with $N(h-1)$ elements by only adding leaves. This property implies
that if we construct an AVL Tree of height $h$ from an initially
empty tree, then the nodes of the final tree can be divided into groups
based on the period they were added as leaves. We note that this property
is a more explicit formulation of the result in \cite{Amani2016},
that proves that an $n$-node AVL tree can be constructed using $n$
inserts. 

The second property shows that if $N_{L}(h)$ is the number of leaves
of the minimum sized AVL tree of height $h$, then $N_{L}(h)=F(h)$.
This property gives a more direct connection between the AVL tree
and the Fibonacci numbers. If we then consider the grouping of nodes
earlier, $N(h)$ is the sum of the first $h$ Fibonacci numbers. Since
the sum of the first $i$ Fibonacci number is known to be equal to
$F(i+2)-1$, the bound on $N(h)$ follows.

For our analysis of build-heap, the key idea is the \textit{heap-merging}
interpretation of the algorithm. First, we treat the array as a forest
containing $n$ heaps. We then view the sift-down operation as a merging
operation to build a larger heap from smaller heaps in the forest.
Finally, the build-heap then becomes a sequence of merge that produces
a single heap. This new interpretation allows us to use a very simple
potential function in terms of the levels of heaps in the forest.
We show that the time taken by a sift-down operation is proportional
to potential loss in the process. 

\section{Analysis of AVL Tree}

Let $T_{h}$ be the AVL tree of height $h$ with minimum number of
nodes. $T_{h}$ can be viewed recursively as a tree containing a root
$r$ with two, possibly empty, subtrees $T_{h-1}$ and $T_{h-2}$. 

Let $N(h)$ be the number of nodes and $N_{L}(h)$ be the number of
leaves of $T_{h}$, respectively. We now establish the relation between
the sizes of $T_{h}$ and $T_{h-1}$.
\begin{lem}
For $h\ge1$, $N(h)=N(h-1)+N_{L}(h)$.
\end{lem}
\begin{proof}
The proof relies on the construction of $T_{h}$ from $T_{h-1}$.
Using the recursive view of $T_{h-1}$, it follows that we must increase
the height of the subtrees of $T_{h-1}$ that contains a leaf as a
child. These are the $T_{0}$ and $T_{1}$ subtrees. A $T_{1}$ subtree
is any height 1 subtree while $T_{0}$ is any height 0 subtree that
is not a subtree of some $T_{1}$. Note that $T_{0}$ contains only
a leaf while $T_{1}$ contains an internal node with one leaf. Increasing
the height of any $T_{0}$ effectively replaces its leaf with a $T_{1}^{'}$
subtree. For the case of a $T_{1}$ subtree, its leaf is replaced
by two subtrees -- $T_{0}^{'}$ and $T_{1}^{'}$. Essentially, the
process replaces all the leaves of $T_{h-1}$ with internal nodes
from all the $T_{1}^{'}$ subtrees then introduces new leaves from
the $T_{0}^{'}$ and $T_{1}^{'}$ subtrees. In effect, we can form
a bijection between the nodes of $T_{h-1}$ and the internal nodes
of $T_{h}$. Therefore, adding the number leaves of $T_{h}$ and the
size of $T_{h-1}$ gives the size of $T_{h}$.
\end{proof}
Lemma 1, when applied repeatedly, suggests that an AVL tree $T_{h}$
can be constructed incrementally beginning with some smaller tree.
We now show that the number of leaves of $T_{h}$ is strongly related
to the Fibonacci numbers,
\begin{lem}
Let $F(i)$ be the $i^{th}$ Fibonacci number, then $N_{L}(h)=F(h)$,
for $h\ge0$.
\end{lem}
\begin{proof}
We show that $N_{L}(h)$ follows the Fibonacci recurrence. For $h\le1$,
direct construction of $T_{0}$ and $T_{1}$ shows that $N_{L}(0)=F(0)$
and $N_{L}(1)=F(1)$. We now show that for $h>1$, $N_{L}(h)=N_{L}(h-1)+N_{L}(h-2)$.
From the proof of Lemma 1, the leaves of $T_{h}$ are the leaves created
after transforming all $T_{0}$ and $T_{1}$ subtrees of $T_{h-1}$.
Let $x_{0}^{'}$ be the leaf added to $T_{h}$ after replacing the
leaf of $T_{0}$ with a $T_{1}^{'}$ subtree. Let $x_{1}^{'}$ and
$y_{1}^{'}$ be the leaves added to $T_{h}$ after replacing the leaf
of $T_{1}$ by $T_{0}^{'}$ and $T_{1}^{'}$ subtrees. Observe that
we can form a bijection between the leaves of $T_{h-1}$ and the $x'$
leaves of $T_{h}$. Therefore, counting all the $x'$ gives $N_{L}(h-1)$.
Also, a bijection can be formed between the roots of each $T_{1}$
and the $y'$ leaves of $T_{h}$. These roots correspond to the leaves
of $T_{h-2}$ since all height 0 nodes (leaves) in $T_{h-2}$ becomes
height 1 nodes (roots of $T_{1}$ subtrees) in $T_{h-1}$. Thus, counting
all $y'$ gives $N_{L}(h-2)$. Since the leaves of $T_{h}$ are the
$x^{'}$ and $y^{'}$ combined, $N_{L}(h)=N_{L}(h-1)+N_{L}(h-2)$
and the statement of the lemma follows.
\end{proof}
We can now easily prove the bound on the size of $T_{h}$,
\begin{cor}
$N(h)=F(h+2)-1$.
\end{cor}
\begin{proof}
Applying Lemma 1 and Lemma 2 repeatedly we have $N(h)=N(0)+\sum_{i=1}^{h}N_{L}(i)=N(0)+\sum_{i=1}^{h}F(i)$.
Since $N(0)=N_{L}(0)=F(0)$, $N(h)=\sum_{i=0}^{h}F(i)$. The claim
follows given that $\sum_{i=0}^{h}F(i)=F(i+2)-1$.
\end{proof}
\begin{thm}
AVL tree has logarithmic height.
\end{thm}
\begin{proof}
From Corollary 2.1, $N(h)=F(h+2)-1$. Since $F(i)\ge\phi^{(i-2)}$
we have $N(h)\ge\phi^{h}-1$. Taking the logarithm of both sides and
letting $n=N(h)$, it follows that the height of an AVL tree is at
most $\log_{\phi}(n+1)$.
\end{proof}

\section{Analysis of Build-Heap}

In our analysis, we shall treat the array as a forest of $n$ trees.
That is, each element of the array is a root of a distinct tree. Since
each tree contains only one node, they can be considered as heaps.
A sift-down can then be interpreted as merging of heaps to produce
one larger heap. Thus, the build-heap algorithm is simply a sequence
of merges to convert the entire forest into a single heap. 

Let the \textit{level} of a heap be equal to the height of the heap
plus one. We let the level of an empty heap be 0. Under this definition,
the forest initially contains $n$ heaps with level 1. Further, a
sift-down merges one heap of level 1, called the \textit{parent heap},
with two heaps with level at most $l$, called \textit{child heaps},
to produce a child heap of level $l+1$. Initially, the parent heaps
are those elements belonging to the upper-half of the array, while
the remaining half are the child heaps. 

We now prove the running-time of build-heap using the potential method:
\begin{thm}
The worst-case running-time of build-heap is O(n).
\end{thm}
\begin{proof}
Let the potential be the total levels of all heaps in the forest.
Since the forest initially contains $n$ heaps of level 1, $\Phi_{0}=n$.
After merging all the heaps, the height of the final heap is $\log_{2}n$,
thus $\Phi_{m}=\log_{2}n+1$. During a sift-down, if the child heaps
have levels $l$ and $l-1$, then in the worst case, the potential
will decrease by $2l$ and then increase by $l+1$. This is so, since
the sift-down will remove three heaps, which has a total level of
$2l$, and then replace them with a heap of level $l+1$. For the
case where the child heaps both have $l$ levels, the decrease in
potential is $2l+1$. 

The actual cost, $a_{i}$, of the $i^{th}$ sift-down operation is
at most $l$ since in the worst case, the root of the parent heap
will be compared to at most $l$ nodes from one of the child heaps.
The amortized cost, $\widehat{a_{i}}$, of a sift-down at time $i$
is:\\

\textit{Case 1}: Child heaps have different levels: $\mathit{\widehat{a_{i}}=a_{i}+\Phi_{i}-\Phi_{i-1}=l-2l+(l+1)=1}$

\textit{Case 2}: Child heaps have same level: $\mathit{\widehat{a_{i}}=a_{i}+\Phi_{i}-\Phi_{i-1}=l-(2l+1)+(l+1)=0}$\\

Since there are only at most $n/2$ parent heaps, the number of sift-down
opeartions, $m$, is at most $n/2$. The amortized cost of a sift-down
is at most 1 therefore, the total amortized cost of the sequence of
sift-down is $n/2$. The worst-case running time of the bottom-up
heap construction is the total actual cost: 

\begin{align*}
\sum_{i=0}^{m}a_{i} & =\sum_{i=0}^{m}\widehat{a_{i}}+\sum_{i=1}^{m}(\Phi_{i-1}-\Phi_{i})\\
 & \le n/2+\Phi_{0}-\Phi_{n/2}\\
 & =n/2+n-\log_{2}n-1\\
 & \le n/2+n\\
 & =1.5n\\
 & =O(n)
\end{align*}
\end{proof}
\printbibliography

\end{document}